\documentclass[11pt,draftcls,onecolumn]{IEEEtran}
\usepackage{amssymb}
\ifCLASSOPTIONcompsoc
  \usepackage[nocompress]{cite}
\else
  \usepackage{cite}
\fi

\usepackage[cmex10]{amsmath}
\usepackage{amsthm}
\usepackage{booktabs}

\newtheorem{theorem}{Theorem}
\newtheorem{lemma}[theorem]{Lemma}
\newtheorem{proposition}[theorem]{Proposition}

\newtheorem{definition}{Definition}

\newtheorem{remark}{\emph{Remark}}
\newtheorem{example}{\emph{Example}}

\newcommand{\bc}{\mathbf{c}}

\date{}

\renewcommand{\baselinestretch}{1.5}

\title{Repeated-root Constacyclic Codes with Optimal Locality}

\author{Wei Zhao, Kenneth W. Shum and Shenghao Yang
 \thanks{W. Zhao is with the School of Science and Engineering, The Chinese University of Hong Kong, Shenzhen, and University of Science and
Technology of China, China (e-mail: zhaowei@cuhk.edu.cn).}
\thanks{K. W. Shum is with the School of Science and Engineering, The Chinese University of Hong Kong, Shenzhen (e-mail: wkshum@cuhk.edu.cn).}
\thanks{S. Yang is with the School of Science and Engineering, The Chinese University of Hong Kong, Shenzhen, Guangdong, 518172, China, the Shenzhen Key Laboratory of IoT Intelligent Systems and Wireless Network Technology, The Chinese University of Hong Kong, Shenzhen,Guangdong, 518172, China, and also with the Shenzhen Research Institute of Big Data, Shenzhen, Guangdong, 518172, China (e-mail:
shyang@cuhk.edu.cn).}
\thanks {This article was presented in part at the 2020 IEEE ISIT.}}

\begin{document}

\maketitle

\begin{abstract}
A code is called a locally repairable code (LRC) if any code symbol is a function of a small fraction of other code symbols. When a locally repairable code is employed in a distributed storage systems, an erased symbol can be recovered by accessing only a small number of other symbols, and hence alleviating the network resources required during the repair process. In this paper we consider repeated-root constacyclic codes, which is a generalization of cyclic codes, that are optimal with respect to a Singleton-like bound on minimum distance. An LRC with the structure of a constacyclic code can be encoded efficiently using any encoding algorithm for constacyclic codes in general. In this paper we obtain optimal LRCs among these repeated-root constacyclic codes. Several infinite classes of optimal LRCs over a fixed alphabet are found. Under a further assumption that the ambient space of the repeated-root constacyclic codes is a chain ring, we show that there is no other optimal LRC.
\end{abstract}
		



\section{Introduction}

\IEEEPARstart Erasure coding techniques can dramatically improve the data availability and reliability of large-scale distributed storage system by adding redundancy. Node failure, however, reduce the effectiveness of erasure coding in terms of both data availability and reliability.
It is shown in \cite{rashmi2013solution} that the node failures in Facebook warehouse cluster occurs very often, and  $98.08\%$ of the failures are single-node failure. The Facebook warehouse cluster is equipped with the $(14, 10)$ Reed-Solomon code. The node repair of Reed-Solomon code is inefficient in that the repair of a failure node need to contact with $10$ other (helper) nodes and download $10$ times of the amount of data stored in the failure node. Hence, node failure may result in network congestion due to data retrieval, and even permanent data loss if too many nodes fail simultaneously.  It is of significant practical interests to design erasure codes that enable efficient node repair.

Various metrics of node repairing costs have been studied in literature, including the \emph{repair-bandwidth} \cite{dimakis2010network}, which is the amount of data traffic required in repairing a failure node, and the \emph{repair locality} \cite{s8,s9}, which is the number of nodes that participate in the repair process. In this paper, we focus on the repair locality. Consider a block code $\mathcal{C}$ of length $n$.
We say a code has \emph{(all-symbol) locality} $r$ if every code symbol can be reconstructed from $r$ other symbols. If an LRC with all-symbol locality is employed in a distributed storage system, we can repair any single node failure from $r$ other surviving nodes. A code with local property is called a \emph{locally repairable code} (LRC). We note that the $r$ helper nodes depend on the index of the failed node, and hence the repair model is different from that in \cite{dimakis2010network}.

LRCs have drawn many research interests recently \cite{s9, s33, s23,s10,s12,s13,s14,s16,s5}.
In~\cite{s9}, a Singleton-like bound is derived for $[n,k,d]$ linear code with locality $r$,
\begin{equation}
  \label{equation1}
  d\leq n-k- \left\lceil\frac{k}{r}\right\rceil+2.
\end{equation}
We note that this bound does not depend on the alphabet size and is valid for nonlinear codes. A bound on the minimum distance of linear LRCs that depends on the alphabet size is given by Cadambe-Mazumdar in~\cite{s33},
 \begin{equation} \label{equation2}
 k\leq \min_{t\in \mathbb{Z}_{+}}\left\{tr+k_{\text{opt}}^{q}(n-t(r+1),d)\right\},
 \end{equation}
where $k_{\text{opt}}^{q}(n',d)$ denotes the largest possible dimension of a $q$-ary linear code of length $n'$ and minimum distance~$d$,  and $\mathbb{Z}_{+}$ denotes the set of all non-negative integers.  A $q$-ary $[n,k,d]$ code with locality $r$ is said to be \emph{$d$-optimal} if it achieves the maximum possible minimum distance among all $q$-ary $[n,k]$ code with locality~$r$.
A $q$-ary $[n,k,d]$ code is said to be \emph{$r$-optimal} if it has locality $r$ and any $q$-ary $[n,k,d]$ code has locality greater than or equal to~$r$.
In this paper, we focus on linear codes that attain equality in the Singleton-like bound in~\eqref{equation1}. These codes are $d$-optimal, and we will refer to them simply as \emph{optimal LRCs}.

For any fixed locality $r$ and minimum distance $d$, the code rate  of optimal LRCs becomes larger as the code length becomes larger (ref. \eqref{equation1}). One of the objectives for constructing LRCs is to achieve a large code length $n$ for a given alphabet size (field size) $q$~\cite{s8,s9,s10,s12,s13,s14,s15}. Based on the classical  MDS conjecture, one could wonder if optimal LRCs can have length longer than $q + 1$. Barg {\em et al}. \cite{s16} constructed optimal LRCs  with $n\approx q^{2}$ for parameters $d=3$ and $2\leq r\leq 4$ using algebraic surface. It is a natural to consider the existence of optimal LRCs with arbitrary code lengths for a given alphabet size~$q$. For distance $d\geq 5$, the code length of optimal LRCs over an alphabet size $q$ is shown to be in the order $O(dq^{3})$~\cite{s4}. The constructions in  \cite{s5} give a class  of optimal LRCs with distance $3$ and $4$ with  unbounded code lengths for a fixed code alphabet size greater than two. Binary LRCs with unbounded length are characterized in~\cite{s6}.

In this paper, we identify the optimal LRCs of length $n>q+1$ within the class of repeated-root constacyclic codes. The result is shown in  Table~\ref{biaoge1}. Constacyclic codes is a generalization of cyclic codes. A formal definition of constacyclic codes is given in the next section. Most of the existing optimal LRCs among cyclic or constacyclic codes in literature, such as \cite{s38,s5,s6,s28,s29,s30,s31,s32,CFXF19,QZ20}, use the zero structure of the codes. As a consequence, all these constructions are confined by the condition that $r+1$ divides $n$, where $r$ is the locality and $n$ is the code length. In contrast, the code length of some of the optimal LRCs obtained in this paper (such as Classes 2,  4 and 9) need not be divisible by $r+1$.


\begin{table}[tb]
  \centering
  \caption{Optimal repeated-root $\lambda$-constacyclic codes over $\mathbb{F}_{p^{m}}$.}
  \label{biaoge1}
  \begin{tabular}{|c|c|c|c|c|c|c| }
    \toprule
		Class & Code & Dimension & Minimum      & Locality & Alphabet  & Remarks\\
		      &  length    &    &  distance &              & size  & \\ \midrule
    1 & $2^{s}$     &$2^{s-1}-1$           &$ 4$&$ 1$&$2^{m}$ &  $s\geq 2$\\
    \midrule
    2 &$p^{s}$    &$p^{s}-p^{s-1}-1$     &$3$&$p-1$ & $p^{m}$ &  $p\geq 3$ and $s\geq2$\\
    \midrule
    3 &$\eta p^{s}$&$\eta(p^{s}-p^{s-\ell-1})$&$2$&$p^{\ell+1}-1$ & $p^{m}$ & $\gcd(\eta,p)=1$, $0\leq \ell\leq s-1$ and $s\geq 2$\\
    \midrule
    4 &$p^{s}$     &$ p^{s}-2$           & $2$&$p^{s}-p^{s-1}-1$ & $ p^{m}$&  $p\geq 3$ and $s\geq2$\\
    \midrule
    5 &$p$     &$p-t$           &$t+1$    &$p-t$ & $p^m$ & $2\leq t\leq p-1$ and $p\geq3$\\
    \midrule
    6 &$p^{s}$& $1$               &$p^{s}$  &$1$ &$p^{m}$&$s\geq2$\\
    \midrule
    7& $2p^s$&$p^s$&$2$ & $1$& $p^m$& $p\geq3$\\
    \midrule
    8& $2p^s$ & $2p^s-p^{s-k-1}$ & $ 2$ & $2p^{k+1}-1$& $p^m$& $0\leq k\leq s-1$, $s\geq2$  and $p\geq3$\\
    \midrule
    9& $2p^s$ & $2p^s-2$ & $2$ & $2p^s-2p^{s-1}-1$ & $p^m$& $s\geq2$ and $p\geq3$ \\
    \midrule
    10 & $2p$ & $p-i$ & $2(i+1)$ & $1$ & $p^m$& $1\leq i\leq p-1$  and $p\geq3$\\
    \bottomrule
   \end{tabular}
  \end{table}

If we take $m=1$ in Class 1, i.e., when the alphabet has size 2, then the LRCs have the same parameters as in an optimal family of LRCs reported in~\cite{s6}. Hence, the LRCs in Class 1 can be regarded as an extension of this optimal family to finite fields of characteristic~2.

The LRCs in Classes 3, 4, 7, 8 and 9 has minimum distance equal 2, which is the same as the minimum distance of a simple parity-check code. In a simple parity-check code, we can compute the lost code symbol by taking the sum of all other code symbols. Using the LRCs in these classes, we do not need to access all the other code symbols in order to repair the erased code symbol, i.e., the locality is much smaller than the code length.

In contrast, the LRCs in Classes 1, 6, 7 and 10 have locality 1, meaning that any erased code symbol can be recovered by reading another code symbol. They have the fastest repair speed, at the expense of larger amount of required redundancy.


The remainder of this paper is organized as follows. In Section~\ref{sec:pre}, we provide
some preliminaries on constacyclic codes and locally repairable codes.
In Section~\ref{sec:pp}, we analyze the locality of repeated-root constacyclic codes of length $\eta p^{s}$ and obtain several classes of optimal constacyclic codes. Some more optimal constacyclic LRCs with length $2p^s$, where $p$ is the characteristic of the field size, are obtained in Section~\ref{sec:eta2}. We prove that under the conditions given in Section~\ref{sec:41}, the optimal LRCs found in section~\ref{sec:pp} are the only repeated-root constacyclic codes that are optimal with respect to the Singleton-like bound. Concluding remarks is given in Section~\ref{sec:conclusion}.

\section{Preliminaries}
\label{sec:pre}
In this section, we present some preliminaries on locally repairable codes and constacyclic codes, and some existing results that we will need later.

Denote by $\mathbb{F}_{q}$ a finite field with cardinality $q=p^{m}$ and $\mathbb{F}_{q}^{n}$ the $n$-dimensional vector space over $\mathbb{F}_{q}$. The set of nonzero elements in $\mathbb{F}_{p^m}$ is denoted by $\mathbb{F}_{p^{m}}^{*}$.  A $q$-ary {\em linear code} of length $n$ is a subspace of $\mathbb{F}_{q}^{n}$.  We index the $n$ code symbols by  $[n]:= \{0,1,\ldots, n-1\}$. The \emph{support} of a vector in $\mathbb{F}_q^n$ is the set of the indices of the nonzero components and the size of the support is called the \emph{weight} of the vector.  The \emph{minimum distance} of a linear code $\mathcal{C}$ is the size of the smallest nonempty support, and is denoted by~$d(\mathcal{C})$.

For each $i\in [n]$, we say that the $i$-th symbol has \emph{locality} $r$ if there exists an index set $I\subset [n]\setminus\{i\}$ of size at most $r$ such that the $i$-th symbol is a deterministic function of the code symbols with indices in~$I$. The code $\mathcal{C}$ is said to have  (all-symbol) \emph{locality} $r$ if every code symbol has locality $r$. A linear code of length $n$, dimension $k$, minimum distance $d$ and locality $r$ is called an $(n,k,d,r)$ locally repairable codes (LRC). Furthermore, we say $\mathcal{C}$ has the {\em minimum locality} $r_{\min}$ if $\mathcal{C}$ has locality $r_{\min}$ but not $r_{\min}-1$.

 For a linear code $\mathcal{C}$ in general of length $n$ over $\mathbb{F}_{q}$, the {\em dual code} $\mathcal{C}^{\perp}$ of $\mathcal{C}$ is the set of vectors in $\mathbb{F}_{q}^{n}$ that are orthogonal to all codewords in $\mathcal{C}$. Recall that two vectors in $\mathbb{F}_q^n$ are said to  be {\em orthogonal} if the inner product between them is zero. We denote the minimum distance of $\mathcal{C}^{\perp}$ by $d^{\perp}$. A code symbol in $\mathcal{C}$ with index $i$ has locality $r$ whenever there exists a codeword in $\mathcal{C}^\perp$ whose support has size at most $r+1$ and contains $i$ as an element.

In this paper we consider LRC that are closed under constacyclic shift. For a nonzero element $\lambda\in \mathbb{F}_{q}$ and a vector
$(c_{0}, c_{1}, \ldots, c_{n-1})\in\mathbb{F}_{q}^{n}$, the \emph{$\lambda$-constacyclic 
shift} $\tau_{\lambda}$ is defined by
\[\tau_{\lambda}(c_{0}, c_{1}, \ldots, c_{n-2}, c_{n-1}) := (\lambda c_{n-1}, c_{0}, c_1, \ldots, c_{n-2}).\]
We use bold letter to denote a vector in $\mathbb{F}_{q}$, e.g., $\bc:=(c_{0}, c_{1}, \ldots, c_{n-1})$. For a positive integer $t$, we define recursively
\begin{equation*}
\tau_{\lambda}^{t}(\bc)=\tau_{\lambda}^{t-1}(\tau_{\lambda}(\bc)),
\end{equation*}
i.e., performing  $\lambda$-constacyclic shift $t$ times on the vector~$\bc$. A $q$-ary linear code $\mathcal{C}$ is called a {\em $\lambda$-constacyclic code} if $\mathcal{C}$ is closed under $\lambda$-constacyclic shift. A $\lambda$-constacyclic code is called a {\em cyclic code} if $\lambda=1$ and a {\em negacyclic code} if $\lambda=-1$.

We identifying a codeword $\bc=(c_{0}, c_{1}, \ldots, c_{n-1})$ with a polynomial
$$c(x)=c_{0}+c_{1}x+\cdots+c_{n-1}x^{n-1}$$
in $\mathbb{F}_q[x]$. We call this polynomial the {\em code polynomial} of codeword $\bc$.
With this association, the codewords in a constacyclic code $\mathcal{C}$ can be regarded as a set of polynomials in the residue ring $\frac{\mathbb{F}_{q}[x]}{\langle x^{n}-\lambda\rangle}$, and a $\lambda$-constacyclic shift can be represented algebraically by multiplying by $x$ and reducing modulo $x^{n}-\lambda$. The following fact is well-known~\cite{s2}:
 A linear code of length $n$ over $\mathbb{F}_{q}$  is $\lambda$-constacyclic if and only if the corresponding code polynomials form an ideal in $\frac{\mathbb{F}_{q}[x]}{\langle x^{n}-\lambda\rangle}$.
Because the residue ring $\frac{\mathbb{F}_{q}[x]}{\langle x^{n}-\lambda\rangle}$ is a principal ideal ring, the ideal in $\frac{\mathbb{F}_{q}[x]}{\langle x^{n}-\lambda\rangle}$ representing the constacyclic code $\mathcal{C}$ is generated by one polynomial. The monic polynomial of smallest degree among the polynomials that generate the ideal is called the {\em generator polynomial}. We note that the generator polynomial is a divisor of $x^{n}-\lambda$. It can be shown that the dual code of a $\lambda$-constacyclic code over $\mathbb{F}_{p^{m}}$ is a $(\lambda^{-1})$-constacyclic code (see e.g.~\cite[Proposition 2.4]{dinh18}).

In this paper we consider constacyclic codes with length $n=\eta p^s$, where $\eta$ is coprime to $p$ and $s$ is a positive integer. The binomial $x^{\eta p^s}-\lambda$ can  be factorized as $(x^\eta - \lambda_0)^{p^s}$, where $\lambda_{0}$ is an element in $\mathbb{F}_{p^{m}}^*$ such that $\lambda_{0}^{p^{s}}=\lambda$. Such a $\lambda_0$ exists because  the Frobenius action is an automorphism.

\begin{definition}
Given $\lambda\in\mathbb{F}_{p^m}^*$, we let $\lambda_0$ be an element in $\mathbb{F}_{p^m}^*$ such that $\lambda_0^{p^s} = \lambda$. For $i=0,1,\ldots, p^{s}$, we use the notation $\mathcal{C}_{i}(\eta,p^s,\lambda_0)$ to denote the $\lambda$-constacyclic code generated by $(x^\eta-\lambda_0)^i$ in $\frac{\mathbb{F}_{p^m}[x]}{ \langle x^{\eta p^{s}}-\lambda\rangle}$,
$$\mathcal{C}_i(\eta,p^s,\lambda_0) := \langle (x^\eta-\lambda_{0})^{i}\rangle /  \langle x^{\eta p^{s}}-\lambda\rangle.$$
When the code parameters $\eta$, $p^s$ and $\lambda_0$ are clear from the context, we will  write $\mathcal{C}_i$ as a short-hand notation of $\mathcal{C}_i(\eta,p^s,\lambda_0)$.
\label{def:C}
\end{definition}

The following proposition provides a method for determining the minimum distance of a repeated-root constacyclic code. It was first proved for repeated-root cyclic codes in~\cite{s3}, but the same conclusion also holds for repeated-root constacyclic codes in general.

\begin{proposition}[\cite{s3}]  \label{prop:Massey}
Let $\mathcal{C}$ be a repeated-root (consta)cyclic code over $\mathbb{F}_q$ of block length $n=p^{\delta}\eta$ with $p$ the characteristic of $\mathbb{F}_{q}$, $\delta\geq 1$ and $\gcd(p,\eta)=1$. Suppose that the generator polynomial $g(x)$ of $\mathcal{C}$ is factorized as
$$
g(x) = \prod_{i=1}^k m_i(x)^{e_i}
$$
where $m_i(x)$ is an irreducible polynomial and $e_i$ is the multiplicity of $m_i(x)$ in $g(x)$, for $i=1,2,\ldots, k$. Let $\bar{\mathcal{C}}_t$ denote the simple-root (consta)cyclic codes of length $\eta$ generated by the product of the irreducible polynomials $m_i(x)$ whose multiplicity $e_i$ in $g(x)$ is strictly larger than~$t$. Then
\begin{equation}
d(\mathcal{C})=\min\{V_{t}\cdot d(\bar{\mathcal{C}}_{t})|\, t=0,1,\ldots, p^{s}-1,\ d(\bar{\mathcal{C}}_{t})>0\},
\label{eq:Massey}
\end{equation}
where
\begin{equation} V_{t}:=\prod_{j=0}^{m-1}(t_{j}+1)
\label{eq:V}
\end{equation}
 and $t_0,\ldots, t_{m-1}$ are the coefficients of the radix-$p$ expansion of $t$, i.e., $t=\sum_{j=0}^{m-1}t_{j}p^{j}$.
\end{proposition}

In the definition of $\bar{\mathcal{C}}_t$ in Prop.~\ref{prop:Massey}, the code $\bar{\mathcal{C}}_t$ is the whole vector space $\mathbb{F}_{p^m}^{\eta}$ when $e_i \leq t $ for all $i$. In which case the minimum distance $d(\bar{\mathcal{C}}_t)$ is 1. The code
$\bar{\mathcal{C}}_t$ could be the zero code consisting of only the zero polynomial. This happens when the product of $m_i(x)$ for all $i=1,2,\ldots, k$ is $x^\eta - \lambda_0$, and it is excluded in the computation of the minimum in \eqref{eq:Massey}.

In the special case when $g(x)$ is the power of a binomial $x^\eta - \lambda_0$, the computations involved in Prop.~\ref{prop:Massey} can be simplified, and we have explicit expression for the minimum distance for the code $C_i(\eta, p^s, \lambda_0)$. This special case is studied in~\cite{Dinh08,OO09,Permouth13,liu17,sharma19}.


\begin{proposition}[\cite{Permouth13} Thm 7.9] \label{prop}
For $0\leq i\leq p^{s}$, the code $\mathcal{C}_{i}(\eta, p^s,\lambda_0)$ with length $n = \eta p^s$ has minimum distance
\begin{align}
d(\mathcal{C}_i) &= \text{min}\{V_{i},V_{i+1},\ldots,V_{p^{s}-1}\} \label{eq:dC1} \\
&= (p-\tau_{\nu-1}+1)p^{s-\nu}, \label{eq:dC2}
\end{align}
where $V_t$ is defined in \eqref{eq:V}, and $\tau_{\nu-1}$ is the most significant nonzero digit in the $p$-ary expansion of $p^s-i$ is $\sum_{j=0}^{\nu-1} \tau_j p^j$.
\end{proposition}

The values of $\nu$ and $\tau_{\nu-1}$ in Prop.~\ref{prop} can be explicitly computed by
\begin{align*}
\nu &= f(p,s,i) := \lfloor \log_p(p^s-i) \rfloor+1 , \\
\tau_{\nu-1} &= g(p,s,i) := \left\lfloor \frac{p^s-i}{p^{\nu-1}}\right\rfloor = \left\lfloor \frac{p^s-i}{p^{f(p,s,i)-1}}\right\rfloor.
\end{align*}

\begin{IEEEproof}
Fix an integer $i$ between 0 and $p^s$. If $x^\eta - \lambda_0$ is irreducible over $\mathbb{F}_{p^m}$, we apply Prop.~\ref{prop:Massey} with $k=1$ and the generator polynomial $g(x)$ factorized as $(x^\eta - \lambda_0)^i$. When $x^\eta - \lambda_0$ is reducible over $\mathbb{F}_{p^m}$, it is factorized as product of irreducible polynomials
$$
x^\eta - \lambda_0 = m_1(x) m_2(x) \cdots m_k(x),
$$
and each irreducible factor has multiplicity 1. In this case we apply Prop.~\ref{prop:Massey} with
$$g(x) = m_1(x)^i m_2(x)^i \cdots m_k(x)^i.$$

For $t< i$, the simple-root constacyclic code $\bar{\mathcal{C}}_i$ in Prop.~\ref{prop:Massey} is the zero code of length $\eta$. For $t\geq i$, $\bar{\mathcal{C}}_i$ is the trivial code $\mathbb{F}_{p^m}^\eta$ with minimum distance 1. By Prop.~\ref{prop:Massey}, the minimum distance of $C_i(\eta, p^s, \lambda_0)$ is the minimum of $V_i$, $V_{i+1}, \ldots, V_{p^s-1}$. This proves equation~\eqref{eq:dC1}. The proof of \eqref{eq:dC2} follows from analyzing the numerical sequence $(V_t)_{t\geq 1}$ and is omitted.
\end{IEEEproof}

\begin{example}
Consider negacyclic code $C_i(1,5^2,-1)$ of length 25 with alphabet $\mathbb{F}_{5}$, generated by $(x+1)^i$ for $i=1,2,\ldots, 26$. To apply Prop.~\ref{prop}, we calculate $V_t$ in \eqref{eq:V}, for $t=0,1,2,\ldots, 24$.
$${\tiny
\begin{array}{|c|ccccccccccccccccccccccccc|} \hline
t&0&1&2&3&4&5&6&7&8&9&10&11&12&13&14&15&16&17&18&19&20&21&22&23&24 \\ \hline
t_0&0&1&2&3&4&0&1&2&3&4&0&1&2&3&4&0&1&2&3&4&0&1&2&3&4 \\
t_1&0&0&0&0&0&1&1&1&1&1&2&2&2&2&2&3&3&3&3&3&4&4&4&4&4 \\ \hline
V_t&1&2&3&4&5&2&4&6&8&10&3&6&9&12&15&4&8&12&16&20&5&10&15&20&25 \\ \hline
d(C_i)&1&2&2&2&2&2&3&3&3&3&3&4&4&4&4&4&5&5&5&5&5&10&15&20&25 \\ \hline
\end{array}
}
$$
The minimum distance of $C_i$ is
$$
d(C_i) = \begin{cases}
1+\lceil i/5 \rceil & \text{ if } 0\leq i \leq 20 , \\
10+5(i-21) & \text{ if } 21 \leq i \leq 24.
\end{cases}
$$
\end{example}

\section{Optimal Constacyclic Codes of Length $\eta p^{s}$}
\label{sec:pp}



In general, given a linear LRC $\mathcal{C}$ with minimum locality $r$, we have $d^{\perp}\leq r+1$, because there is a codeword of weight at most $r+1$ in $\mathcal{C}^{\perp}$. 
The next theorem shows that we have equality $d^{\perp}= r+1$ when $d^\perp \geq 2$.

\begin{lemma}\label{lemma2}
Let $\mathcal{C}$ be a constacyclic code of length $n$ with $d^\perp \geq 2$. Then the minimum locality of $\mathcal{C}$ is equal to $d^{\perp}-1$.
\end{lemma}

\begin{IEEEproof}
There exists a codeword $\bc$ of $\mathcal{C}^{\perp}$ with weight $d^{\perp}$. Since $\mathcal{C}$ is constacyclic, $\mathcal{C}^\perp$ is also constacyclic. Any  constacyclic shift of $\bc$ is a codeword of $\mathcal{C}^{\perp}$ with weight $d^\perp$. Thus, all code symbols has locality $d^{\perp}-1$. The minimum locality of $\mathcal{C}$ cannot be strictly less than $d^\perp-1$, because it would produce a codeword in $\mathcal{C}^\perp$ with weight strictly less than $d^\perp$. Thus $d^{\perp}-1$ is the minimum locality of $\mathcal{C}$.
\end{IEEEproof}

As a result, determining the minimum locality of $\mathcal{C}_i$ is equivalent to determining the minimum distance of the dual code $\mathcal{C}^\perp$. For constacyclic code, the Singleton-like bound in~\eqref{equation1} becomes
\begin{equation}
d \leq n - k - \left\lceil \frac{k}{d^\perp-1}\right\rceil + 2.
\label{eq:Singleton2}
\end{equation}

Using the characterization of minimum distance in Prop.~\ref{prop}, we can obtain several classes of optimal LRCs in Theorem~\ref{thm 5-2}.

\begin{theorem}\label{thm 5-2}
Let $\lambda_{0}$ be an element in $\mathbb{F}_{p^{m}}^*$. There are six classes of optimal repeated-root $\lambda$-constacyclic LRCs over $\mathbb{F}_{p^{m}}$ in $\frac{\mathbb{F}_{p^m}[x]}{\langle (x^{\eta} - \lambda_0)^{p^s} \rangle}$:
 \begin{enumerate}

\item[1.] \label{t2:e1}  $(2^{s},2^{s-1}-1,4,1)$ LRC $\mathcal{C}_{2^{s-1}+1}(1,2^s,\lambda_0)$ with $p=2$ and $s\geq 2$;

 \item [2.]\label{t2:e3}  $(p^{s},p^{s}-p^{s-1}-1,3,p-1)$ LRC $\mathcal{C}_{p^{s-1}+1}(1, p^s, \lambda_0)$,  where $s\geq 2$ and $p$ is an odd prime;

 \item [3.] \label{t2:e2} $(\eta p^{s},\eta (p^{s}-p^{s-\ell-1}),2,p^{\ell+1}-1)$ LRC $ \mathcal{C}_{p^{s-\ell-1}}(\eta,p^s,\lambda_0)$ with $\eta \geq 1$, $0\leq \ell\leq s-1$ and $s\geq 2$;

 \item [4.] \label{t2:e4} $(p^{s},p^{s}-2,2,p^{s}-p^{s-1}-1)$ LRC $\mathcal{C}_2(1,p^s,\lambda_0)$, where $s\geq 2$ and $p$ is an odd prime;


\item [5.]\label{t2:e7} $(p,p-t,t+1,p-t)$ LRC $\mathcal{C}_t(1,p,\lambda_0)$, where $2\leq t\leq p-1$ and $p$ is an odd prime.

 \item [6.]\label{t2:e6} $(p^{s},1,p^{s},1)$ LRC $\mathcal{C}_{p^s-1}(1,p^s,\lambda_0)$, where $s\geq 2$;

\end{enumerate}
When $x^\eta- \lambda_0$ is irreducible over $\mathbb{F}_{p^m}$ there is no other optimal LRC among the repeated-root constacyclic codes in $\frac{\mathbb{F}_{p^m}[x]}{\langle (x^{\eta} - \lambda_0)^{p^s} \rangle}$.
\end{theorem}

\begin{remark} The codes in the Class 5 and Class 6 are trivial LRCs. Class 5 consists of maximal-distance separable (MDS) codes with $r=k$. When $t=p-1$, the code $\mathcal{C}_{p-1}(1,p,\lambda_0)$ is the trivial code that contains all vectors in $\mathbb{F}_p^p$. The code is Class 6 is the repetition code of length $p^s$ and dimension~1.
\end{remark}

\begin{remark}
If $x^{\eta }-\lambda_{0}$ is reducible over $\mathbb{F}_{p^{m}}$, the problem of determining all the generator polynomials of $\lambda$-constacyclic codes of length $\eta p^{s}$ becomes less tractable. Anyway, we still have the six classes of repeated-root constacyclic LRCs in Theorem~\ref{thm 5-2}.
\end{remark}

\begin{proof}[Proof of Theorem~\ref{thm 5-2}]
We check that each of the six classes attains the Singleton-like bound with equality. In each case we apply Prop.~\ref{prop} to the dual code $\mathcal{C}_i^\perp = \mathcal{C}_{p^s- i}(\eta, p^s, \lambda_0^{-1})$ to obtain the locality of $\mathcal{C}_i$. In the followings, denote the length of the $p$-ary expansion of an integer $a$ by $\nu(a)$, and the most significant nonzero digit by $\tau(a)$.

{\bf Class 1.} The dimension of $\mathcal{C}_{2^{s-1}+1}(1,2^s,\lambda_0)$ is $k=2^s -(2^{s-1}+1) = 2^{s-1}-1$. By Prop.~\ref{prop}, the minimum distance of $\mathcal{C}_{2^{s-1}+1}(1,2^s,\lambda_0)$ and $\mathcal{C}_{2^{s-1}+1}(1,2^s,\lambda_0)^\perp$ are
\begin{align*}
d &= 2\cdot 2^{s-\nu(2^s-(2^{s-1}+1))} =2\cdot 2^{s-\nu(2^{s-1}-1)} = 4,  \\
d^\perp & =2\cdot 2^{s-\nu(2^s-(2^{s-1}-1))} = 2\cdot 2^{s-\nu(2^{s-1}+1)}= 2,
\end{align*}
respectively. The locality of $\mathcal{C}_{2^{s-1}+1}(1,2^s,\lambda_0)$ is $d^\perp -1 = 1$.
It achieves the Singleton-like bound
$$
d = 4 = 2^s - (2^{s-1}-1 ) - \left\lceil \frac{2^{s-1}-1 }{1} \right\rceil + 2
$$
with equality.

{\bf Class 2.} Suppose $p$ is an odd prime and $s\geq 2$. The repeated-root constacyclic codes in this class has generator polynomial $(x-\lambda_0)^{p^{s-1}+1}$ and length $p^s$. From the $p$-ary expansion of
$$p^s - (p^{s-1}+1) = \sum_{j=0}^{s-2}(p-1)p^j +(p-2)p^{s-1},$$
we can calculate,
$$\tau(p^s - (p^{s-1}+1)) = p-2, \quad \nu(p^s - (p^{s-1}+1)) = s$$
and
$$d =  (p - (p-2)+1) p^{s-s} = 3$$
by Prop.~\ref{prop}. Also, we  have
$$
d^\perp = (p-\tau(p^{s-1}+1)+1) p^{s-\nu(p^{s-1}+1)} = (p-1+1)p^{s-s} =p.
$$
The locality in the class is therefore equal to $d^\perp-1=p-1$. One can verified that the Singleton-like bound
\begin{align*}
d &\leq p^s - (p^s-(p^{s-1}+1)) -  \left\lceil \frac{p^s-(p^{s-1}+1)}{p-1} \right\rceil +2\\
&= p^{s-1}+1 - \left\lceil \frac{p^{s-1}(p-1)-1}{p-1} \right\rceil + 2\\
&= 3
\end{align*}
is met with equality.

{\bf Class 3.} Let $\eta \geq 1$, $s\geq 2$ and $0\leq \ell \leq s-1$. The dimension of $\mathcal{C}_{p^{s-\ell-1}}(\eta,p^s,\lambda_0)$ is $k=\eta(p^s - p^{s-\ell-1})$. From
\begin{gather*}
\tau(p^s-p^{s-\ell-1}) = p-1 , \quad \nu(p^s-p^{s-\ell-1})  =s , \\
\tau(p^{s-\ell-1})=1, \quad \nu(p^{s-\ell-1})=s-\ell ,
\end{gather*}
we obtain the minimum distance of $\mathcal{C}_{p^{s-\ell-1}}(\eta,p^s,\lambda_0)$ and its dual by \eqref{eq:dC2},
\begin{align*}
d &=  (p-(p-1)+1) \cdot p^{s-s} = 2 \\
d^\perp & =  (p-1+1) \cdot p^{s-(s-\ell)} = p^{\ell+1}.
\end{align*}
The locality is thus $d^\perp-1 = p^{\ell+1}-1$. We check that the Singleton-like bound says that the minimum distance is bounded by 2,
\begin{align*}
d &\leq \eta p^s - \eta(p^s - p^{s-\ell-1}) - \left\lceil \frac{\eta(p^s - p^{s-\ell-1})}{p^{\ell+1}-1} \right\rceil + 2  \\
&= \eta p^{s-\ell-1} - \left\lceil \frac{\eta p^{s-\ell-1}(p^{\ell+1} - 1)}{p^{\ell+1}-1} \right\rceil + 2 .
\end{align*}

{\bf Class 4.} Suppose $p$ is an odd prime and $s\geq 2$. The codes in this class are generated by $(x-\lambda_0)^2$. We first verify that
\begin{gather*}
\tau(p^s-2) = p-1, \quad \nu(p^s-2) = s,\\
\tau(2) = 2, \quad \nu(2) = 1.
\end{gather*}
The minimum distance and the dual minimum distance are
\begin{align*}
d       &=  (p-(p-1)+1) \cdot p^{s-s} = 2, \\
d^\perp & =  (p-2+1) \cdot p^{s-1} = (p-1)p^{s-1},
\end{align*}
respectively, and the locality is $r=(p-1)p^{s-1}-1 = p^s - p^{s-1}-1$. The Singleton-like bound is met with equality,
$$
d \leq p^s - (p^s-2) - \left\lceil \frac{p^s-2}{p^s - p^{s-1}-1} \right\rceil + 2 = 2 - 2 + 2 = 2.
$$

{\bf Class 5.} The code length is $p$ and the generator polynomial is $(x-\lambda_0)^{t+1}$, where $p$ is an odd prime and $1\leq t \leq p-2$. We apply Prop.~\ref{prop} and determine the minimum distance and locality,
\begin{align*}
d &= (p-(p-t)+1)p^{1-1} = t+1 \\
d^\perp &= (p-t+1)p^{1-1} -1 = p-t+1.
\end{align*}
One can check that the locality $d^\perp - 1=p-t$ is equal to the dimension $k$ of $\mathcal{C}_t(1,p,\lambda_0)$, and it is an MDS code, i.e., $d = n-k+1 = p - (p-t) + 1 = t+1$. Therefore, this is an optimal LRC.

{\bf Class 6.} The codes in class 6 is the repetition codes of length $p^s$. This is a class of trivial LRCs.

\smallskip

If $x^\eta - \lambda_0$ is irreducible, a divisor of $x^{\eta p^s}-\lambda$ can be written as $(x^\eta - \lambda_0)^i$ for $i=0,1,\ldots, p^s$, and  the residue ring $\frac{\mathbb{F}_{p^{m}}[x]}{\langle x^{\eta p^{s}}-\lambda\rangle}$ is a chain ring. The ideals of  $\frac{\mathbb{F}_{p^{m}}[x]}{\langle x^{\eta p^{s}}-\lambda\rangle}$ are precisely the ideals generated by $(x^\eta -\lambda_0)^i$ for $i=0,1,2,\ldots, p^s$. The proof that there is no other optimal LRCs when $x^\eta - \lambda_0$ is irreducible is relegated to the appendix.
\end{proof}

\begin{remark}
 It is shown in \cite{s4} that the minimum distance of optimal linear LRCs with unbounded length has to be less than $5$.
The authors of \cite{s5} give a class of optimal LRCs
of distance $3$ and $4$ with unbounded length over fixed finite fields whose locality is greater than or equal to $ 3$.  The LRCs in Class~1 to Class~4 in Theorem \ref{thm 5-2} are examples of LRCs with ununbounded length and minimum distances 2, 3, and 4.
 \end{remark}

\begin{remark} Optimal binary LRCs are completely characterized in \cite{s6}. Class 1 in Theorem~\ref{thm 5-2}, belongs to one of the infinite families of LRCs in~\cite{s6} with minimum distance $d=4$ and locality $r=1$. Class~2 belongs to another infinite family in~\cite{s6} with $d=2$ and $r|k$. This shows that some optimal binary LRCs have the structure of repeated-root cyclic codes.
\end{remark}

\begin{example}
We apply Theorem~\ref{thm 5-2} with length $n = 3p^s$ and alphabet size $p^m=4$. Let $\omega$ be a primitive cube root of unity in $\mathbb{F}_4$. Class 2 in Theorem~\ref{thm 5-2} gives $s-1$ optimal LRCs $\mathcal{C}_{p^{s-k-1}}(3,p^s,\omega)$, for $k=1,2,\ldots, s-1$, and Class 5 gives one more optimal LRC $\mathcal{C}_{p^{s-1}}(3,2^s, \omega)$. Since $x^{3}-\omega$ is irreducible over $\mathbb{F}_4$. These $s$ LRCs are all the optimal LRCs that are ideals in $\frac{\mathbb{F}_4[x]}{\langle (x^3-\omega)^{2^s} \rangle}$. For example, when $s=4$, the four optimal LRCs $\mathcal{C}_{2^{3-k}}(3,16,\omega)$, for $k=0,1,2,3$, have parameters $(48,24,2,1)$, $(48,36,2,3)$, $(48,42,2,7)$ and $(48,45,2,15)$, respectively.
\end{example}

\begin{example}
Consider cyclic codes of code length $64$ over $\mathbb{F}_{2}$, according to Theorem \ref{thm 5-2}, there are exactly $8$ optimal cyclic codes of length $64$ over $\mathbb{F}_{2}$. 
The parameters of these $8$ optimal LRCs are $(64,63,2,63)$, $(64,62,2,31)$, $(64,60,2,15)$, $(64,56,2,7)$, $(64,48,2,3)$, $(64,32,2,1)$, $(64,31,4,1)$ and $(64,1,64,1)$. 
\end{example}

\section{Optimal Constacyclic LRCs of Length $2 p^{s}$}
\label{sec:eta2}
In this section we consider optimal constacyclic codes of length $2p^{s}$ over $\mathbb{F}_{p^{m}}$, which are the ideals of
$$\frac{\mathbb{F}_{p^{m}}[x]}{\langle (x^{2}-\lambda_0)^{p^{s}} \rangle},
$$
where $p$ is an odd prime and $\lambda_{0}$ is an element of $\mathbb{F}_{p^{m}}^*$. 
Note that $x^{2}-\lambda_{0}$ is irreducible over $\mathbb{F}_{p^{m}}$ if and only if $\lambda_{0}$ is a quadratic non-residue in $\mathbb{F}_{p^{m}}^*$. Actually, there are $\frac{p^{m}-1}{2}$ quadratic non-residues in $\mathbb{F}_{p^{m}}^{*}$.  When $x^2-\lambda_0$ is irreducible over $\mathbb{F}_{p^m}$, only the third class of constacyclic codes in Theorem~\ref{thm 5-2} with parameters $(2 p^{s},2 (p^{s}-p^{s-k-1}),2,p^{k+1}-1)$, for $k=0,1,2,\ldots, s-1$, are optimal LRCs.

\begin{example}
Suppose $p^{m}\equiv 3\bmod 4$. Then $-1$ is a quadratic non-residue in $\mathbb{F}_{p^{m}}$, and $x^2+1$ is irreducible over $\mathbb{F}_{p^m}$. For $0\leq k\leq s-1$, the code
$$\mathcal{C}_{p^s-k-1}(2,p^s,-1) = \langle (x^{2}+1)^{p^s-k-1}\rangle \subseteq \mathbb{F}_{p^{m}}[x]/\langle x^{2p^{s}}+1\rangle$$
  are optimal negacyclic LRCs over $\mathbb{F}_{p^{m}}$ with parameters $(2 p^{s},2 (p^{s}-p^{s-k-1}),2,p^{k+1}-1)$. For example, when the code length is $54 = 2\cdot 27$, we have three optimal negacyclic LRCs with generator polynomials, $x^{18}+1$, $x^6+1$ and $x^2+1$ in the ring $\mathbb{F}_{27}[x]/(x^{54}+1)$, and their code parameters are $(54,36,2,2)$, $(54,48,2,8)$ and $(54,52,2,26)$, respectively.
\end{example}

In the rest of this section we consider $\lambda_{0}$ that is a quadratic residue in $\mathbb{F}_{q}^{*}$, so that $x^2-\lambda_0$ is factorized as the product of two linear factors, 
$x^{2}-\lambda_{0}=(x-\delta)(x+\delta)$ for some $\delta \in \mathbb{F}_{p^m}^*$.
We denote $\lambda_0^{p^s}$ by $\lambda$. The $\lambda$-constacyclic codes of length $2p^s$ over $\mathbb{F}_{p^m}$ are ideals in the form
 $$\langle (x-\delta)^{i}(x+\delta)^{j}\rangle$$
in the residue ring $\frac{\mathbb{F}_{p^m}[x]}{\langle x^{2p^s}-\lambda\rangle}$,
for $0\leq i,j\leq p^s$.

\begin{definition} For odd prime $p$, a quadratic residue $\lambda_0$ with square root $\delta$ in $\mathbb{F}_{p^m}$, and $0\leq i,j\leq p^s$, let $\mathcal{C}_{i,j}(p^s,\lambda_{0})$ be the constacyclic code  over $\mathbb{F}_{p^m}$  associated with the ideal $\langle (x-\delta)^{i}(x+\delta)^{j}\rangle$ in $\mathbb{F}_{p^m}[x]/((x^2-\lambda_0)^{p^s})$.
\label{def2}
\end{definition}

Without loss of generality, we can assume $i<j$ in Definition~\ref{def2}. For the case of $i=j$, the generator polynomial of $\mathcal{C}_{i,i}(2p^s,\lambda_{0})$ can be represented as $(x^2-\delta)^i$, where $\delta$ is a quadratic non-residue, and this reduces to the scenario treated in Theorem~\ref{thm 5-2}.


\begin{proposition}[\cite{Permouth13} Remark 7.11] \label{prop2}
Let $p$ be an odd prime. For $0\leq i<j\leq p^{s}$, the code $\mathcal{C}_{i,j}(2p^s,\lambda_0)$ with length $n = 2 p^s$ has minimum distance
\begin{align}
d(\mathcal{C}_i) &= \text{min}\{2V_{i},2V_{i+1},\ldots,2V_{j-1},V_j,V_{j+1},\ldots, V_{p^s-1}\} \label{eq:d1} \\
&= \text{min}\{2(p-\tau_{\nu_1-1}+1)p^{s-\nu_1}, (p-\tau_{\nu_2-1}+1)p^{s-\nu_2}\}\label{eq:d2}
\end{align}
where $V_t$ is defined in \eqref{eq:V}, and
$\nu_1$ and $\tau_{\nu_1-1}$, $\nu_2$ and $\tau_{\nu_1-2}$ are computed  by
\begin{align*}
\nu_1 &= f(p,s,i) := \lfloor \log_p(p^s-i) \rfloor+1 , \\
\nu_2 &= f(p,s,j) := \lfloor \log_p(p^s-j) \rfloor+1 , \\
\tau_{\nu_1-1} &= g(p,s,i) := \left\lfloor \frac{p^s-i}{p^{\nu_1-1}}\right\rfloor = \left\lfloor \frac{p^s-i}{p^{f(p,s,i)-1}}\right\rfloor,\\
\tau_{\nu_2-1} &= g(p,s,j) := \left\lfloor \frac{p^s-j}{p^{\nu_2-1}}\right\rfloor = \left\lfloor \frac{p^s-j}{p^{f(p,s,j)-1}}\right\rfloor.
\end{align*}
\end{proposition}

\begin{IEEEproof}
Fix integers $i$, $j$ between 0 and $p^s$ and $i<j$. we apply Prop.~\ref{prop:Massey} with $k=2$ and the generator polynomial $g(x)$ factorized as $(x - \delta)^i(x+\delta)^j$. For $t<i$, the simple-root constacyclic code $\bar{\mathcal{C}}_i$ in Prop.~\ref{prop:Massey} is the zero code of length~$2$. For $i\leq t< j$, $\bar{\mathcal{C}}_i$ is the code $\langle x+\delta\rangle$ of length $2$ with minimum distance~$2$. For $t\geq j$, $\bar{\mathcal{C}}_i$ is the trivial code $\mathbb{F}_{p^m}^2$ with minimum distance $1$. By Prop.~\ref{prop:Massey}, the minimum distance of $\mathcal{C}_{i,j}(2p^s, \lambda_0)$ is the minimum of $2V_{i},2V_{i+1},\ldots,2V_{j-1},V_j,V_{j+1},\ldots, V_{p^s-1}$.
This proves equation~\eqref{eq:d1}. The proof of \eqref{eq:d2} follows from analyzing the numerical sequence $(V_t)_{t\geq 1}$ and is omitted.
\end{IEEEproof}

\begin{theorem}\label{thm 6-1}
Let $\lambda_{0}$ be a quadratic residue in $\mathbb{F}_{p^{m}}^*$ and $p$ be an odd prime. There are four classes of optimal repeated-root $\lambda$-constacyclic LRCs over $\mathbb{F}_{p^{m}}$ in $\frac{\mathbb{F}_{p^m}[x]}{\langle (x^{2} - \lambda_0)^{p^s} \rangle}$:
 \begin{enumerate}
 \item[1.] \label{t4:e1} $(2p^s,p^s,2,1)$ LRC $\mathcal{C}_{0,p^s}(2p^s,\lambda_0)$ with $s\geq1$;
 \item[2.] \label{t4:e2} $(2p^s,2p^s-p^{s-k-1},2,2p^{k+1}-1)$ LRC $\mathcal{C}_{0,p^{s-k-1}}(2p^s,\lambda_0)$, where $0\leq k\leq s-1$ and $s\geq 2$;
 \item[3.] \label{t4:e3} $( 2p^s,2p^s-2,2,2p^s-2p^{s-1}-1)$ LRC $\mathcal{C}_{0,2}(2p^s,\lambda_0)$ with $s\geq 2$;
 \item[4.] \label{t4:e4} $( 2p,p-i,2(i+1),1)$ LRC $\mathcal{C}_{i,p}(2p,\lambda_0)$, where $1\leq i\leq p-1$.
\end{enumerate}
\end{theorem}

\begin{IEEEproof}
We check that the LRCs in each of the four classes attain the Singleton-like bound with equality. Prop.~\ref{prop2} is invoked to obtain  the locality of $\mathcal{C}_{i,j}$ via the dual code. In the followings, denote the length of the $p$-ary expansion of an integer $a$ by $\nu(a)$, and the most significant nonzero digit by $\tau(a)$. We denote $\tau(0)=-\infty$. According to Proposition \ref{prop2}, we have the minimum distances of $\mathcal{C}_{i,j}$ and $C_{i,j}^\perp$ are
\begin{align}
  d & =\text{min}\{2(p-\tau(p^s-i)+1)p^{s-\nu(p^s-i)}, (p-\tau(p^s-j)+1)p^{s-\nu(p^s-j)}\} \label{th8:e1}\\
  d^{\perp} & = \text{min}\{2(p-\tau(j)+1)p^{s-\nu(j)}, (p-\tau(i)+1)p^{s-\nu(i)}\} \label{th8:e2}
\end{align}

{\bf Class 1.} The dimension of $\mathcal{C}_{0,p^s}(2p^s,\lambda_0)$ is $2p^s -p^s = p^s$. By \eqref{th8:e1} and \eqref{th8:e2}, the minimum distance of $\mathcal{C}_{0,p^s}(2p^s,\lambda_0)$ and $\mathcal{C}_{0,p^s}(2p^s,\lambda_0)^\perp$ are
\[
   d =\text{min}\{2(p-\tau(p^s)+1)p^{s-\nu(p^s)}, (p-\tau(0)+1)p^{s-\nu(0)}\}=2,
\]
\[
 d^{\perp}  = \text{min}\{2(p-\tau(p^s)+1)p^{s-\nu(p^s)}, (p-\tau(0)+1)p^{s-\nu(0)}\}=2,
\]
respectively. The locality of $\mathcal{C}_{0,p^s}(2p^s,\lambda_0)$ is $d^\perp -1 = 1$.
It achieves the Singleton-like bound
$$
d = 2 = 2p^s - p^s - \left\lceil \frac{p^s}{1} \right\rceil + 2
$$
with equality.

{\bf Class 2.} Suppose $s\geq 2$. The repeated-root constacyclic codes in this class has generator polynomial $(x+\delta)^{p^{s-k-1}}$ and length $2p^s$. From the $p$-ary expansion of
$$p^s - p^{s-k-1} = \sum_{j=s-k-1}^{s-1}(p-1)p^j,$$
we can calculate,
\begin{align*}
  \tau(p^s - p^{s-k-1}) = p-1,&  \quad \nu(p^s - p^{s-k-1}) = s, \\
  \tau(p^{s-k-1})= 1,&  \quad \nu(p^{s-k-1}) = s-k.
\end{align*}
According to \eqref{th8:e1} and \eqref{th8:e2}, we have
 the minimum distance of $\mathcal{C}_{0,p^{s-k-1}}(2p^s,\lambda_0)$ and $\mathcal{C}_{0,p^{s-k-1}}(2p^s,\lambda_0)^\perp$ are
 \begin{align*}
   d = & \text{min}\{2(p-\tau(p^s)+1)p^{s-\nu(p^s)}, (p-\tau(p^s - p^{s-k-1})+1)p^{s-\nu(p^s - p^{s-k-1})}\}=2, \\
   d^{\perp}  =  & \text{min}\{2(p-\tau(p^{s-k-1})+1)p^{s-\nu(p^{s-k-1})}, (p-\tau(0)+1)p^{s-\nu(0)}\}=2p^{k+1},
 \end{align*}
respectively. The locality of $\mathcal{C}_{0,p^{s-k-1}}(2p^s,\lambda_0)$ is $d^\perp -1 = 2p^{k+1}-1$.
One can verified that the Singleton-like bound
\begin{align*}
d &\leq 2p^s - (2p^s - p^{s-k-1}) -  \left\lceil \frac{2p^s - p^{s-k-1}}{2p^{k+1}-1} \right\rceil +2\\
&= p^{s-k-1}- \left\lceil \frac{(p^{s-k-1})(2p^{k+1}-1)}{2p^{k+1}-1)} \right\rceil + 2\\
&= 2
\end{align*}
is met with equality.

{\bf Class 3.}
Let $p$ be an odd prime. The dimension of $\mathcal{C}_{0,2}(2p^s,\lambda_0)$ is $k=2p^s-2$. From the $p$-ary expansion of
$$p^s - 2 = \sum_{j=1}^{s-1}(p-1)p^j+p-2,$$
we can calculate,
\begin{align*}
  \tau(p^s - 2) = p-1,&  \quad \nu(p^s -2) = s, \\
  \tau(2)= 2,&  \quad \nu(2) = 1.
\end{align*}
According to \eqref{th8:e1} and \eqref{th8:e2}, we have
 the minimum distance of $\mathcal{C}_{0,2}(2p^s,\lambda_0)$ and $\mathcal{C}_{0,2}(2p^s,\lambda_0)^\perp$ are
 \begin{align*}
   d = & \text{min}\{2(p-\tau(p^s)+1)p^{s-\nu(p^s)}, (p-\tau(p^s - 2)+1)p^{s-\nu(p^s - 2)}\}=2, \\
   d^{\perp}  =  & \text{min}\{2(p-\tau(2)+1)p^{s-\nu(2)}, (p-\tau(0)+1)p^{s-\nu(0)}\}=2p^{s}-2p^{s-1},
 \end{align*}
respectively. The locality of $\mathcal{C}_{0,p^{s-k-1}}(2p^s,\lambda_0)$ is $d^\perp -1 = 2p^{s}-2p^{s-1}-1$.
We check that the Singleton-like bound says that the minimum distance is bounded by 2,
\begin{align*}
d &\leq 2 p^s - (2p^s -2) - \left\lceil \frac{2p^s -2}{2p^{s}-2p^{s-1}-1} \right\rceil + 2  \\
&= 2- \left\lceil 1+\frac{2p^{s-1}-1}{2(p-1)p^{s-1}-1} \right\rceil + 2 .
\end{align*}
{\bf Class 4.} The dimension of $\mathcal{C}_{i,p}(2p,\lambda_0)$ is $p-i$. By \eqref{th8:e1} and \eqref{th8:e2}, the minimum distance of $\mathcal{C}_{i,p}(2p,\lambda_0)$ and $\mathcal{C}_{i,p}(2p,\lambda_0)^\perp$ are
\[
   d =\text{min}\{2(p-\tau(p-i)+1)p^{s-\nu(p-i)}, (p-\tau(0)+1)p^{s-\nu(0)}\}=2(i+1),
\]
\[
 d^{\perp}  = \text{min}\{2(p-\tau(p)+1)p^{s-\nu(p)}, (p-\tau(i)+1)p^{s-\nu(i)}\}=2,
\]
respectively. The locality of $\mathcal{C}_{i,p}(2p,\lambda_0)$ is $d^\perp -1 = 1$.
It achieves the Singleton-like bound
$$
d = 2i+2 = 2p-(p-i) - \left\lceil \frac{p-i}{1} \right\rceil + 2
$$
with equality.
\end{IEEEproof}

\section{Irreducible Binomial $x^{\eta }-\lambda_{0}$  over $\mathbb{F}_{p^{m}}$}
\label{sec:41}

In Theorem~\ref{thm 5-2}, we have a complete characterization of optimal LRCs among the repeated-root constacyclic codes of length $\eta p^s$ when $x^\eta - \lambda_0$ is irreducible over $\mathbb{F}_{p^m}$. In this section we determine all the values of $\eta$ such that $x^\eta - \lambda_0$ that is irreducible over $\mathbb{F}_{p^{m}}$.


For an element $a\in \mathbb{F}_{p^{m}}^{*}$, the \emph{order} of $a$ is defined as the smallest positive integer $e$ such that $a^{e}=1$. Since the multiplicative subgroup of $\mathbb{F}_{p^m}$ is cyclic, the order of any nonzero element in $\mathbb{F}_{p^m}$ is a divisor of $p^m-1$. Conversely, given any divisor of $p^m-1$, we can find an element in $\mathbb{F}_{p^m}^*$ whose order is equal to the divisor. An element in $\mathbb{F}_{p^m}^*$ with order $p^m-1$ is called a {\em primitive element} of $\mathbb{F}_{p^m}$.

The following proposition gives a criterion for the irreducibility of  $x^{\eta}-\lambda_{0}$ over $\mathbb{F}_{p^{m}}$.

\begin{proposition}[\cite{s1} Thm 3.75]  \label{t1}
Let $\eta\geq 2$ be an integer and $a\in\mathbb{F}_{p^{m}}^{*}$. Then the binomial $x^\eta-a$ is irreducible in $\mathbb{F}_{p^{m}}[x]$ if and only if the following two conditions are satisfied:
\begin{enumerate}
\item[(i)] each prime factor of $\eta$ divides the order $e$ of $a$ in $\mathbb{F}_{p^{m}}^{*}$, but not $(p^{m}-1)/e$;
\item[(ii)]  $p^{m}\equiv 1\bmod 4$ if $\eta\equiv 0\bmod 4$.
\end{enumerate}
\end{proposition}



Suppose $x^\eta - \lambda_0$ is irreducible over $\mathbb{F}_{p^m}$.
From condition (i) in Proposition~\ref{t1}, each prime factor of of $\eta$ should divide $p^m-1$ because the order of $\lambda_0$ is a divisor of $p^m-1$. Consider the prime factorization of $p^m-1$,
\begin{equation}
p^m-1 = 2^{\ell_0} p_1^{\ell_1} p_2^{\ell_2} \cdots p_s^{\ell_s},
\label{eq:pm-1}
\end{equation}
where $p_1,\ldots, p_s$ are distinct odd prime divisors of $p^m-1$,  $\ell_i$ are positive integers for $i=1,2,\ldots, s$, and $\ell_0\geq 0$. If we factorize $\eta$, the prime factorization of $\eta$ can be written as
\begin{equation}
\eta = 2^{e_0} p_1^{e_1} p_2^{e_2} \cdots p_s^{e_s},
\label{eq:eta_factorization}
\end{equation}
with $e_i \geq 0$ for $i=0,1,2,\ldots, s$.

The next proposition is an immediate consequence of Prop.~\ref{t1}.

\begin{proposition}
Let $\xi$ be a primitive element in $\mathbb{F}_{p^m}$. Suppose $p^m-1$ is factorized as in \eqref{eq:pm-1}. The binomial $x^\eta - \xi $ is irreducible over $\mathbb{F}_{p^m}$ if and only if $\eta$ can be expressed as in \eqref{eq:eta_factorization} with $e_i\geq 0$ for $i=1,2,\ldots, s$ and
\begin{enumerate}
\item[(i)] $e_0 \geq 0$ if $\ell_0 \geq 2$;

\item[(ii)] $e_0 \in \{0,1\}$ if $\ell_0=1$;

\item[(iii)] $e_0=0$ if $\ell_0 = 0$.
\end{enumerate}
\label{prop:factor_primitive}
\end{proposition}

\begin{example}\label{ex 5-1}
Consider the quaternary field $\mathbb{F}_{4}$ with four elements $\{0,1,\omega,\overline{\omega} \}$. Here $\overline{\omega}=1+\omega$ and $\omega^2+\omega+1=0$. The element $\omega$ is a primitive element of $\mathbb{F}_4$. By Prop.~\ref{prop:factor_primitive}, $x^\eta-\omega$ is irreducible over $\mathbb{F}_4$ if and only if $\eta = 3^\mu$ for some positive integer $\mu$.
\end{example}

The next proposition generalizes Prop.~\ref{prop:factor_primitive}.

\begin{proposition}\label{pro}
Let $\mathbb{F}_{p^{m}}$ be a finite field and $\omega\in\mathbb{F}_{p^m}^*$ is an element with order $f$. Suppose that the prime factorization of $p^m-1$ is  given as in~\eqref{eq:pm-1}. For notational convenience, we define $p_0:=2$. Let $S$ be a subset of $\{0,1,\ldots, s\}$ defined by
$$
S := \{ i\in[s] :\,  p_i^{\ell_i} \text{ divides } f \},
$$
Then $x^\eta - \omega$ is irreducible over $\mathbb{F}_{p^m}$ if and only if $\eta$ can be written as
$$
\eta = \prod_{i\in S} p_i^{e_i}
$$
with $e_i\geq 0$ for $i\in S\setminus\{0\}$, and
\begin{enumerate}
\item[(i)] $e_0 \geq 0$ if $p^m-1 \equiv 0 \bmod 4$;
\item[(ii)] $e_0 \in \{0,1\}$ if $p^m-1 \equiv 2 \bmod 4$;
\item[(iii)] $e_0 = 0$ if $p^m-1$ is odd.
\end{enumerate}
When $S$ is empty, the product $\prod_{i\in S} p_i^{e_i}$ is equal to 1 by convention.
\end{proposition}

We remark that when the set $S$ in Prop.~\ref{pro} is not empty then there are infinitely many $\eta$ such that the binomial $x^\eta - \lambda_0$ is irreducible over $\mathbb{F}_{p^m}$.



\begin{example}
Consider a finite field $\mathbb{F}_{64}$ and let $\alpha$ be a primitive element in $\mathbb{F}_{64}$. The size of the multiplicative subgroup is odd and can be factorized as $63=3^2\cdot 7$. The element $\omega = \alpha^{21}$ has order~3. By Prop.~\ref{pro}, the only binomial in the form $x^\eta - \omega$ that is irreducible over $\mathbb{F}_{64}$ is $x-\omega$, because the corresponding set $S$ in Prop.~\ref{pro} is empty. On the other hand, the element $\omega'=\alpha^7$ has order 9. The binomial $x^{3^\mu}-\omega'$ is irreducible over $\mathbb{F}_{64}$ for $\mu\geq 0$.
\end{example}

The next example is an example with $p^m-1 = 2 \bmod 4$.

\begin{example}
Let $\xi$ be a primitive element of $\mathbb{F}_{7^3}= \mathbb{F}_{343}$. We have the factorization
$$7^3-1=2\cdot 3^2 \cdot 19.
$$
The element $\xi^{19}$ has order 18.
 By Prop.~\ref{pro}, the binomial $x^\eta - \xi^{19}$ is irreducible over $\mathbb{F}_{27}$ if and only if $\eta = 2(13)^\mu$ or $\eta=(13)^\mu$ for $\mu\geq 0$.
\end{example}

\section{Concluding Remarks}
\label{sec:conclusion}
Constructing optimal linear codes with long code length and small locality is an important code construction problem for locally repairable code. In this paper, we found several classes of optimal LRCs within the constacyclic codes of length $\eta p^{s}$. In general, when $x^n-\lambda_0$ is reducible, the enumeration of constacyclic codes in $\mathbb{F}_{p^m}[x]/((x^\eta - \lambda_0)^{p^s})$ is unsolved for arbitrary value of $\eta$. We refer the readers to \cite{chen14,dinh12, zhao18} and the references therein.

A linear code with $(r,\delta)$-locality ensures a code tolerate multiple node failures while preserving locality. We have an upper bound on the minimum distance of linear codes with $(r,\delta)$-locality
 $$d\leq n-k+1-\left(\Big\lceil \frac{k}{r}\Big\rceil-1\right)(\delta-1).$$
Characterizing constacyclic codes that are optimal $(r,\delta)$-LRC is an interesting future research direction.

Some other directions for future work include:

\begin{enumerate}
\item Characterize the locality of the $\lambda$-constacyclic codes of length $\eta p^{s}$ over $\mathbb{F}_{q}$ when $x^{\eta}-\lambda_{0}$ is reducible over $\mathbb{F}_{q}$, where $\eta$ is a positive integer coprime to $p$ and $\lambda$, $\lambda_{0}$ are elements in $\mathbb{F}_{q}$ satisfying $\lambda_{0}^{p^{s}}=\lambda$.
\item Find out more optimal repeated-root constacyclic codes.
 \item Constructing more linear codes that are both $d$-optimal and $r$-optimal with large distance and small locality.
\end{enumerate}



\appendix
In the appendix, we will show that when $x^{\eta}-\lambda_0$ is irreducible over $\mathbb{F}_{p^m}$, there are only these optimal LRCs we shown in Theorem \ref{thm 5-2}. Note that under the condition of $x^{\eta}-\lambda_0$ being irreducible over $\mathbb{F}_{p^m}$, there are $p^s+1$ different $\lambda$-constacyclic codes of length $\eta p^s$ over $\mathbb{F}_{p^m}$ which are $\mathcal{C}_{i}(\eta,p^s,\lambda_0)=\langle (x^\eta-\lambda_0)^i\rangle$ for $0\leq i\leq p^s$. We will check all the cases for $0\leq i\leq p^s$.

Note that $\mathcal{C}_{i}(\eta,p^s,\lambda_0)$ is an optimal LRC for each $0\leq i\leq p^{s}$ if and only if the equality $$\eta i=\lceil\frac{\eta (p^{s}-i)}{r_{i}}\rceil+d_{i}-2$$ holds.
Verify that $\mathcal{C}_{0}(\eta,p^s,\lambda_0)=\langle 1\rangle$ and $\mathcal{C}_{p^{s}}(\eta,p^s,\lambda_0)=\langle 0\rangle$ are not optimal LRCs. We discuss $i\in \{1,\ldots, p^{s}-1\}$ in four cases:
\begin{itemize}
\item $1\leq i\leq p^{s-1}-1$;
  \item $i=p^{s-1}$;
  \item $p^{s-1}+1\leq i\leq p^{s}-p^{s-1}$;
  \item $p^{s}-p^{s-1}+1\leq i\leq p^{s}-1$.
\end{itemize}

According to Theorem \ref{thm 5-2}, $\mathcal{C}_{p^{s-k-1}}(\eta,p^s,\lambda_0)$ and $\mathcal{C}_2(1,p^s,\lambda_0)$ are optimal LRCs. The following lemma shows there are only these two classes of optimal LRCs when $1\leq i\leq p^{s-1}-1$.
\begin{lemma}\label{lm1}
Let $s\geq 2$ be an integer, $\eta$ be a positve integer coprime to $p$. For $1\leq i\leq p^{s-1}-1$, there exist optimal LRCs $\mathcal{C}_i(\eta,p^s,\lambda_0)$ only if
\begin{itemize}
  \item $i=p^{s-k-1}$, where $1\leq k\leq s-1$;
  \item $i=2$, $\eta=1$ and $p$ is an odd prime.
\end{itemize}
 \end{lemma}

\begin{IEEEproof}
We divide the interval $[1,p^{s-1}-1]$ into $(p-1)(s-1)$ disjoint parts:
\begin{equation*}
p^{s-k}-tp^{s-k-1}\leq i\leq p^{s-k}-(t-1)p^{s-k-1}-1,
\end{equation*}
 where $1\leq t\leq p-1$ and $1\leq k\leq s-1$. 
By Proposition \ref{prop}, $\mathcal{C}_{i}$ has parameters $[\eta p^{s},\eta(p^{s}-i),2]$ and locality $(t+1)p^{k}-1$. Suppose that $\mathcal{C}_{i}$ is an optimal LRC. Then we have
\begin{equation*}
\eta i=\lceil\frac{\eta(p^{s}-i)}{(t+1)p^{k}-1}\rceil.
\end{equation*}
Since $\eta$ is a positive integer, it follows that
\begin{equation}\label{eq1}
 i=\lceil\frac{p^{s}-i}{(t+1)p^{k}-1}\rceil.
\end{equation}
Hence
\begin{equation*}
i=\lceil \frac{p^{s-k}}{t+1}\rceil.
\end{equation*}

 If $t+1$ divides $p^{s-k}$, then $t=p-1$ since $1\leq t\leq p-1$. It follows that $i=p^{s-k-1}$, where $1\leq k\leq s-1$. 

 If $t+1$ does not divide $p^{s-k}$,
then $1\leq t\leq p-2$, and this implies $p\geq 3$ and $t+1\leq p-1<p$, hence
\begin{equation*}
\frac{p^{s-k}}{t+1}<p^{s-k}-tp^{s-k-1}.
\end{equation*}
Combining with $i=\lceil \frac{p^{s-k}}{t+1}\rceil$ and $i\geq p^{s-k}-tp^{s-k-1}$, we obtain
\begin{equation*}
\frac{p^{s-k}}{t+1}+1>p^{s-k}-tp^{s-k-1},
\end{equation*}
which implies that
\begin{equation*}
tp^{s-k-1}(p-t-1)<t+1.
\end{equation*}
Therefore, it forces $p^{s-k-1}=1$ and $p-t-1=1$. Hence $k=s-1$, $t=p-2$ and
\begin{equation*}
i=\lceil \frac{p^{s-(s-1)}}{p-2+1}\rceil=2.
\end{equation*}
Note that the equality (\ref{eq1}) holds in the assumption that $i$ existence. Assume that $\mathcal{C}_{2}$ is an optimal LRC with parameters $[\eta p^{s},\eta(p^{s}-2),2]$ and locality $p^{s}-p^{s-1}-1$.
Then
\begin{equation}\label{equa:1}
  2\eta=\lceil \frac{\eta(p^{s}-2)}{p^{s}-p^{s-1}-1}\rceil
\end{equation}
We claim that (\ref{equa:1}) holds only if $\eta=1$. When $\eta>1$, it follows from $p\geq3$ that
$\frac{p^{s-1}-1}{p^{s}-p^{s-1}-1}<\frac{1}{2}$ and $\lceil \frac{\eta(p^{s-1}-1)}{p^{s}-p^{s-1}-1}\rceil<\frac{\eta}{2}+1\leq \eta$, which implies that
\begin{equation*}
  2\eta>\lceil \frac{\eta(p^{s}-2)}{p^{s}-p^{s-1}-1}\rceil.
\end{equation*}
\end{IEEEproof}

For the case of $i=p^{s-1}$, according to Theorem \ref{thm 5-2}, $\mathcal{C}_{p^{s-1}}(\eta,p^s,\lambda_0)$ is optimal LRC. In the interval of $[p^{s-1}+1, p^{s}-p^{s-1}]$, we know $\mathcal{C}_{p^{s-1}+1}(1,p^s,\lambda_0)$ and $\mathcal{C}_{t+1}(1,p,\lambda_0)$ are optimal LRCs. The following lemma shows there are only these two classes of optimal LRCs if $p^{s-1}+1\leq i\leq p^{s}-p^{s-1}\}$.


\begin{lemma}\label{lm3}
Let $s$ be a positive integer, $\eta$ be an integer coprime to $p$. For $p^{s-1}+1\leq i\leq p^{s}-p^{s-1}\}$, there exist optimal LRCs $\mathcal{C}_i(\eta,p^s,\lambda_0)$ only if
\begin{itemize}
  \item $i=p^{s-1}+1$, $s\geq2$, $\eta=1$ and $p$ is an odd prime;
  \item $i=t+1$, $s=1$, $\eta=1$ and $p$ is an odd prime.
\end{itemize}
\end{lemma}
\begin{IEEEproof}
We divide the interval $[p^{s-1}+1,p^{s}-p^{s-1}]$ into $p-2$ disjoint parts:
\begin{equation*}
t p^{s-1}+1\leq i\leq (t+1)p^{s-1},
\end{equation*}
where $1\leq t\leq p-2$.
For each part $[t p^{s-1}+1,(t+1)p^{s-1}]$, we first consider the case of $i=(t+1)p^{s-1}$. According to Proposition \ref{prop}, $\mathcal{C}_{(t+1)p^{s-1}}$ has parameters $[\eta p^{s},\eta(p^{s}-(t+1)p^{s-1}),t+2]$ and locality $p-t-1$.
Suppose that $\mathcal{C}_{(t+1)p^{s-1}}$ is an optimal LRC. By the Singleton-like bound \eqref{equation1} it deduces $\eta p^{s-1}=1$ and hence $s=1$, $\eta=1$ and $i=t+1$.

For $tp^{s-1}+1\leq i\leq (t+1)p^{s-1}-1$, we know that $\mathcal{C}_{i}$ has parameters $[\eta p^{s},\eta(p^{s}-i),t+2]$ and
locality $p-t$.  Note that $s>1$ since $tp^{s-1}+1\leq(t+1)p^{s-1}-1$. 
Assuming that $\mathcal{C}_{i}$ is an optimal LRC. Then
\begin{equation*}
\eta i=\lceil \frac{\eta(p^{s}-i)}{p-t}\rceil +t.
\end{equation*}
Thus we obtain
\begin{equation}\label{8}
 \frac{\eta p^{s}+pt-t^{2}}{\eta(p-t+1)}\leq i<\frac{\eta p^{s}+pt-t^{2}}{\eta(p-t+1)}+\frac{p-t}{\eta(p-t+1)}.
\end{equation}
If $t=1$, we have
\begin{equation*}
p^{s-1}+\frac{p-1}{\eta p}\leq i<p^{s-1}+\frac{2(p-1)}{\eta p},
\end{equation*}
Note that $\frac{2(p-1)}{\eta p}>1$ since $i$ is an integer and $0<\frac{p-1}{\eta p}<1$, which forces $\eta =1$ and $p\geq 3$. Therefore, $i=p^{s-1}+1$. 
\\If $1<t\leq p-2$, it implies $p>3$. 
One can verify that
\begin{equation*}
tp^{s-1}+1>\frac{\eta p^{s}+pt-t^{2}}{\eta(p-t+1)}+\frac{p-t}{\eta(p-t+1)}.
\end{equation*}
Together with (\ref{8}), we obtain $i<tp^{s-1}+1$, which contradicts to $tp^{s-1}+1\leq i\leq (t+1)p^{s-1}-1$. Thus $\mathcal{C}_{i}$ is not an optimal LRC in this case.
\end{IEEEproof}

When $p^{s}-p^{s-1}+1\leq i\leq p^{s}-1$, Theorem \ref{thm 5-2} shows two classes of optimal LRCs which are $\mathcal{C}_{2^{s-1}+1}(1,2^s,\lambda_0)$ and $\mathcal{C}_{p^{s}-1}(1,p^s,\lambda_0)$. The following lemma shows there are only these two classes of optimal LRCs when $p^{s}-p^{s-1}+1\leq i\leq p^{s}-1$.

\begin{lemma}\label{lm4}
Let $s\geq 2$ be a positive integer, $\eta$ be a positive integer coprime to $p$. For $p^{s}-p^{s-1}+1\leq i\leq p^{s}-1$, there exist optimal LRCs $\mathcal{C}_i(\eta,p^s,\lambda_0)$ only if
\begin{itemize}
  \item $i=2^{s-1}+1$, $\eta=1$ and $p=2$;
  \item $i=p^{s}-1$, $\eta=1$.
\end{itemize}
\end{lemma}
\begin{IEEEproof}
  Consider the intervals $[p^{s}-p^{s-k}+(t-1)p^{s-k-1}+1,p^{s}-p^{s-k}+tp^{s-k-1}]$,
 where $1\leq t\leq p-1$ and $1\leq k\leq s-1$. Note that
 \begin{equation*}
   \bigcup\limits_{\mbox{\tiny$\begin{array}{c}1\leq t\leq p-1\\1\leq k\leq s-1 \end{array}$}}\{p^{s}-p^{s-k}+(t-1)p^{s-k-1}+1\leq i\leq  p^{s}-p^{s-k}+tp^{s-k-1}\}=\{p^{s}-p^{s-1}+1\leq i\leq p^{s}-1\}.
 \end{equation*}
According to Prop.~\ref{prop:Massey}, $\mathcal{C}_{i}$ has parameter $[\eta p^{s},\eta(p^{s}-i),(t+1)p^{k}]$. Note that
$1\leq p^{s}-i\leq p^{s-1}-1$,
it follows from  Proposition \ref{prop} that $\mathcal{C}_{i}$ has locality
$1$. Suppose that  $\mathcal{C}_{i}$ is an optimal LRC. Then by the Singleton-like bound~\eqref{equation1},
we have
\begin{equation}\label{e1}
i=\frac{p^{s}}{2}+\frac{(t+1)p^{k}}{2}-1.
\end{equation}
 Since $i\geq p^{s}-p^{s-k}+(t-1)p^{s-k-1}+1$, \[\frac{p^{s}}{2}+\frac{(t+1)p^{k}}{2\eta}-\frac{1}{\eta}\geq p^{s}-p^{s-k}+(t-1)p^{s-k-1}+1.\]
It follows that
\begin{equation}\label{inequ:1}
\frac{1}{2}(p^{k}-2)(p^{s-k}-\frac{t+1}{\eta})+(t-1)(p^{s-k-1}-1)+(t-\frac{t}{\eta})\leq 0.
\end{equation}
 Check that for prime $p$, $1\leq k\leq s-1$ and $1\leq t\leq p-1$,
 \begin{equation*}
   \frac{1}{2}(p^{k}-2)(p^{s-k}-\frac{t+1}{\eta})\geq 0
 \end{equation*}
 and
 \[(t-1)(p^{s-k-1}-1)\geq 0.\]

 If $\eta>1$, then $t-\frac{t}{\eta}>0$, which implies that the inequality (\ref{inequ:1}) is impossible. Therefore, the code $\mathcal{C}_{i}$ for $p^{s}-p^{s-1}+1\leq i\leq p^{s}-1$ cannot be optimal LRC in the case of $\eta>1$.

 If $\eta=1$, the inequality in (\ref{inequ:1}) implies
\begin{equation}\label{eq in le 1}
  \frac{1}{2}(p^{k}-2)(p^{s-k}-t-1)= 0
\end{equation}
 and
 \begin{equation}\label{eq in le 2}
  (t-1)(p^{s-k-1}-1)= 0.
 \end{equation}
We consider the equalities (\ref{eq in le 1}) and (\ref{eq in le 2}) in the following two cases.

Case 1. If $p^{k}-2=0$. Then $p=2$ and $k=1$.  Since $1\leq t\leq p-1$, it forces $t=1$ and then $(t-1)(p^{s-k-1}-1)= 0$. According to (\ref{e1}), $i=2^{s-1}+1$.

Case 2. If $(p^{s-k}-t-1)= 0$. Since $2\leq t+1\leq p$, it forces
\[p^{s-k}=t+1=p.\]
Then $s=k+1$ and $t=p-1$. According to (\ref{e1}), $i=p^{s}-1$.
\end{IEEEproof}

According to the previous three lemmas, we obtain the second part of Theorem \ref{thm 5-2}, i.e., when $x^\eta- \lambda_0$ is irreducible over $\mathbb{F}_{p^m}$ there is no other optimal LRC among the repeated-root constacyclic codes in $\frac{\mathbb{F}_{p^m}[x]}{\langle (x^{\eta} - \lambda_0)^{p^s} \rangle}$. 
\ifCLASSOPTIONcaptionsoff
  \newpage
\fi

\bibliographystyle{IEEEtran}
\bibliography{IEEEabrv,LRC}

\begin{thebibliography}{10}
\providecommand{\url}[1]{#1}
\csname url@samestyle\endcsname
\providecommand{\newblock}{\relax}
\providecommand{\bibinfo}[2]{#2}
\providecommand{\BIBentrySTDinterwordspacing}{\spaceskip=0pt\relax}
\providecommand{\BIBentryALTinterwordstretchfactor}{4}
\providecommand{\BIBentryALTinterwordspacing}{\spaceskip=\fontdimen2\font plus
\BIBentryALTinterwordstretchfactor\fontdimen3\font minus
  \fontdimen4\font\relax}
\providecommand{\BIBforeignlanguage}[2]{{%
\expandafter\ifx\csname l@#1\endcsname\relax
\typeout{** WARNING: IEEEtran.bst: No hyphenation pattern has been}%
\typeout{** loaded for the language `#1'. Using the pattern for}%
\typeout{** the default language instead.}%
\else
\language=\csname l@#1\endcsname
\fi
#2}}
\providecommand{\BIBdecl}{\relax}
\BIBdecl

\bibitem{rashmi2013solution}
K.~Rashmi, N.~B. Shah, D.~Gu, H.~Kuang, D.~Borthakur, and K.~Ramchandran, ``A
  solution to the network challenges of data recovery in erasure-coded
  distributed storage systems: A study on the facebook warehouse cluster,'' in
  \emph{5th {USENIX} Workshop on Hot Topics in Storage and File Systems
  (HotStorage 13)}, 2013.

\bibitem{dimakis2010network}
A.~Dimakis, P.~Godfrey, Y.~Wu, M.~Wainwright, and K.~Ramchandran, ``Network
  coding for distributed storage systems,'' \emph{{IEEE} Trans. on Information
  Theory}, vol.~56, no.~9, pp. 4539--4551, Sep. 2010.

\bibitem{s8}
J.~Han and L.~A. Lastras{-}Montano, ``Reliable memories with subline
  accesses,'' in \emph{{IEEE} Int. Symp. on Inf. Theory}, Nice, Jun. 2007, pp.
  2531--2535.

\bibitem{s9}
P.~Gopalan, C.~Huang, H.~Simitci, and S.~Yekhanin, ``On the locality of
  codeword symbols,'' \emph{{IEEE} Trans. on Information Theory}, vol.~58,
  no.~11, pp. 6925--6934, 2012.

\bibitem{s33}
V.~R. Cadambe and A.~Mazumdar, ``Bounds on the size of locally recoverable
  codes,'' \emph{{IEEE} Trans. on Information Theory}, vol.~61, no.~11, pp.
  5787--5794, 2015.

\bibitem{s23}
A.~Wang, Z.~Zhang, and D.~Lin, ``Bounds for binary linear locally repairable
  codes via a sphere-packing approach,'' \emph{{IEEE} Trans. on Information
  Theory}, vol.~65, no.~7, pp. 4167--4179, 2019.

\bibitem{s10}
N.~Silberstein, A.~S. Rawat, O.~O. Koyluoglu, and S.~Vishwanath, ``Optimal
  locally repairable codes via rank-metric codes,'' in \emph{{IEEE} Int. Symp.
  on Inf. Theory}, Istanbul, Jul. 2013, pp. 1819--1823.

\bibitem{s12}
I.~Tamo and A.~Barg, ``A family of optimal locally recoverable codes,''
  \emph{{IEEE} Trans. on Information Theory}, vol.~60, no.~8, pp. 4661--4676,
  2014.

\bibitem{s13}
L.~Jin, L.~Ma, and C.~Xing, ``Construction of optimal locally repairable codes
  via automorphism groups of rational function fields,'' \emph{{IEEE} Trans. on
  Information Theory}, vol.~66, no.~1, pp. 210--221, 2020.

\bibitem{s14}
X.~Li, L.~Ma, and C.~Xing, ``Optimal locally repairable codes via elliptic
  curves,'' \emph{{IEEE} Trans. on Information Theory}, vol.~65, no.~1, pp.
  108--117, 2019.

\bibitem{s16}
A.~Barg, K.~Haymaker, E.~W. Howe, G.~L. Matthews, and A.~V{\'a}rilly-Alvarado,
  ``Locally recoverable codes from algebraic curves and surfaces,'' in
  \emph{Algebraic Geometry for Coding Theory and Cryptography}.\hskip 1em plus
  0.5em minus 0.4em\relax Springer, 2017, pp. 95--127.

\bibitem{s5}
Y.~Luo, C.~Xing, and C.~Yuan, ``Optimal locally repairable codes of distance 3
  and 4 via cyclic codes,'' \emph{{IEEE} Trans. on Information Theory},
  vol.~65, no.~2, pp. 1048--1053, 2019.

\bibitem{s15}
Z.~Zhang, J.~Xu, and M.~Liu, ``Constructions of optimal locally repairable
  codes over small fields,'' \emph{SCIENTIA SINICA Mathematica}, vol.~47,
  no.~11, pp. 1607--1614, 2017.

\bibitem{s4}
V.~Guruswami, C.~Xing, and C.~Yuan, ``How long can optimal locally repairable
  codes be?'' \emph{{IEEE} Trans. on Information Theory}, vol.~65, no.~6, pp.
  3662--3670, 2019.

\bibitem{s6}
J.~Hao, S.-T. Xia, K.~W. Shum, B.~Chen, F.-W. Fu, and Y.~Yang, ``Bounds and
  constructions of locally repairable codes: Parity-check matrix approach,''
  \emph{{IEEE} Trans. on Information Theory}, vol.~66, no.~12, pp. 7465--7474,
  Dec. 2020.

\bibitem{s38}
I.~Tamo and A.~Barg, ``A family of optimal locally recoverable codes,''
  \emph{{IEEE} Trans. on Information Theory}, vol.~60, no.~8, pp. 4661--4676,
  2014.

\bibitem{s28}
P.~Tan, Z.~Zhou, H.~Yan, and U.~Parampalli, ``Optimal cyclic locally repairable
  codes via cyclotomic polynomials,'' \emph{{IEEE} Communications Letters},
  vol.~23, no.~2, pp. 202--205, 2019.

\bibitem{s29}
Z.~Sun, S.~Zhu, and L.~Wang, ``Optimal constacyclic locally repairable codes,''
  \emph{{IEEE} Communications Letters}, vol.~23, no.~2, pp. 206--209, 2019.

\bibitem{s30}
A.~Beemer, R.~Coatney, V.~Guruswami, H.~H. Lopez, and F.~Pi{\~{n}}ero,
  ``Explicit optimal-length locally repairable codes of distance 5,'' in
  \emph{56th Annual Allerton Conference on Communication, Control, and
  Computing, Allerton, Monticello, IL, USA, October 2-5}, 2018, pp. 800--804.

\bibitem{s31}
B.~Chen, S.-T. Xia, J.~Hao, and F.-W. Fu, ``Constructions of optimal cyclic
  $(r, \delta)$ locally repairable codes,'' \emph{{IEEE} Trans. on Information
  Theory}, vol.~64, no.~4, pp. 2499--2511, 2018.

\bibitem{s32}
I.~Tamo, A.~Barg, S.~Goparaju, and A.~R. Calderbank, ``Cyclic {LRC} codes and
  their subfield subcodes,'' in \emph{{IEEE} Int. Symp. on Inf. Theory}, Hong
  Kong, Jun. 2015, pp. 1262--1266.

\bibitem{CFXF19}
B.~Chen, W.~Fang, S.-T. Xia, and F.-W. Fu, ``Constructions of optimal
  $(r,\delta)$ locally repairable codes via constacyclic codes,'' \emph{{IEEE}
  Trans. on Information Theory}, vol.~65, no.~8, pp. 5753--5763, Aug. 2019.

\bibitem{QZ20}
J.~Qian and L.~Zhang, ``New optimal cyclic locally recoverable codes of length
  $n = 2(q + 1)$,'' \emph{{IEEE} Trans. on Information Theory}, vol.~66, no.~1,
  pp. 233--239, Jan. 2020.

\bibitem{s2}
W.~C. Huffman and V.~Pless, \emph{Fundamentals of Error-Correcting
  Codes}.\hskip 1em plus 0.5em minus 0.4em\relax Cambridge University Press,
  2003.

\bibitem{dinh18}
H.~Q. Dinh, Y.~Fan, H.~Liu, X.~Liu, and S.~Sriboonchitta, ``On self-dual
  constacyclic codes of length $p^s$ over $\mathbb{F}_{p^m} + u
  \mathbb{F}_{p^m}$,'' \emph{Discrete Mathematics}, vol. 341, no.~2, pp.
  324--335, 2018.

\bibitem{s3}
G.~Castagnoli, J.~L. Massey, P.~A. Schoeller, and N.~von Seemann, ``On
  repeated-root cyclic codes,'' \emph{{IEEE} Trans. on Information Theory},
  vol.~37, no.~2, pp. 337--342, 1991.

\bibitem{Dinh08}
H.~Q. Dinh, ``On the linear ordering of some classes of negacyclic and cyclic
  codes and their distance distributions,'' \emph{Finite Fields and Their
  Applications}, no.~14, pp. 22--40, 2008.

\bibitem{OO09}
H.~\"{O}zadam and F.~\"{O}zbudak, ``A note on negacyclic and cyclic codes of
  length $p^s$ over a finite field of characteristic $p$,'' \emph{Advances in
  Mathematics of Communications}, vol.~3, no.~3, pp. 265--271, 2009.

\bibitem{Permouth13}
S.~R. {L\'{o}pez-Permouth}, H.~\"{O}zadam, F.~\"{O}zbudak, and S.~Szabo,
  ``Polycyclic codes over {G}alois rings with applications to repeated-root
  constacyclic codes,'' \emph{Finite Fields and Their Applications}, no.~19,
  pp. 16--38, 2013.

\bibitem{liu17}
H.~Liu and Y.~Maouche, ``Some repeated-root constacyclic codes over {G}alois
  rings,'' \emph{{IEEE} Trans. on Information Theory}, vol.~63, no.~10, pp.
  6247--6255, Oct. 2017.

\bibitem{sharma19}
A.~Sharma and T.~Sidana, ``On the structure and distances of repeated-root
  constacyclic codes of prime power lengths over finite commutative chain
  rings,'' \emph{{IEEE} Trans. on Information Theory}, vol.~65, no.~2, pp.
  1072--1084, Feb. 2019.

\bibitem{s1}
R.~Lidl and H.~Niederreiter, \emph{Finite Fields}, 2nd~ed., ser. Encyclopedia
  of Math. and Its Applications.\hskip 1em plus 0.5em minus 0.4em\relax
  Cambridge University Press, 1997, vol.~20.

\bibitem{chen14}
B.~Chen, H.~Q. Dinh, and H.~Liu, ``Repeated-root constacyclic codes of length $
  lp^{s}$ and their duals,'' \emph{Discrete Applied Mathematics}, vol. 177, pp.
  60--70, 2014.

\bibitem{dinh12}
H.~Q. Dinh, ``Repeated-root constacyclic codes of length $2p^{s}$,''
  \emph{Finite Fields and Their Applications}, vol.~18, no.~1, pp. 133--143,
  2012.

\bibitem{zhao18}
W.~Zhao, X.~Tang, and Z.~Gu, ``Constacyclic codes of length $kl^{m}p^{n}$ over
  a finite field,'' \emph{Finite Fields and Their Applications}, vol.~52, pp.
  51--66, 2018.

\end{thebibliography}

\end{document}